\documentclass{elsarticle}
\usepackage{amssymb}
\usepackage{amsmath}
\usepackage{amsthm}
\usepackage{mathtools}
\DeclarePairedDelimiter\abs{\lvert}{\rvert}
\newtheorem{theorem}{Theorem}
\newtheorem{lemma}{Lemma}
\def\E{\mathbb{E}}
\def\Beta{\textrm{Beta}}
\newtheorem{cor}{Corollary}

\begin{document}
\begin{frontmatter}
\title{Moments of the multivariate Beta distribution}
\author{Feng Zhao}

 \begin{abstract}
In this paper, we extend Beta distribution to 2 by 2 matrix and
give the analytical
formula for its moments. Our analytical formula can be used to analyze the asymptotic behavior of Beta distribution
for 2 by 2 matrix.
 \end{abstract}

\begin{keyword}
	multivariate Beta distribution \sep higher moments
\end{keyword}
\end{frontmatter}
\section{Introduction}
Moments of probability distribution are an important topic in statistics. Given the moment sequence,
the probability distribution is unique under some mind conditions.
To prove the convergence of random variables,  we can prove the convergence of its moment sequences instead.
To accomplish such goal, the analytical form of moments is a prerequisite. The techniques to compute moments
for different distributions differ. In this article, we focus
on the Beta distribution of 2 by 2 matrix.

David introduces an extension of
multivariate extension for Beta distribution,
denoted as $\mathbf{B}(\alpha, \beta; I_p)$ (see \cite{david1981}).
It is a random $p\times p$  symmetric matrix $W$ whose density
function is given by
\begin{align}
p(w) &= \frac{1}{B_p(\alpha, \beta)}\abs{I-w}^{\alpha-\frac{p+1}{2}}
\abs{w}^{\beta-\frac{p+1}{2}} \textrm{ where } w, I-w \in S_{p,p}^{++}
\label{eq:distr}\\
B_p(\alpha, \beta) &= \int_{w,I-w \in S_{p,p}^{+} }\abs{I-w}^{\alpha-\frac{p+1}{2}}
\abs{w}^{\beta-\frac{p+1}{2}}dw \textrm{ where } \alpha, \beta > \frac{p-1}{2}
\end{align}
$B_p(\alpha, \beta)$ is called the multivariate Beta function (see \cite{siegel_1935}); 
$\abs{W}$ is the determinant of matrix $W$ and $S_{p,p}^{++}$ is the 
collection of positive
definite matrix.
When $p=1$, the distribution reduces to normal Beta distribution for
$0<x<1$.

This extension may have useful applications in multivariate statistical
problems but little is known about the analytical property of such extension.
For example, it is unknown whether the moments $\E[f(W)]$ can be written in concise form,
where $f$ is a monomial about the positive-definite matrix $W$.

Konno has derived the formula of the moment up to second order (see \cite{konno_1988}).
In this paper, we focus on the case $p=2$ and deduce the analytical form of 
moments for $\mathbf{B}(\alpha, \beta; I_2)$.
This formula
includes the expectation and variance, which are the first and second
order moment respectively. Our moments formula, as
far as we know, is novel and can be used directly in the computation
related with multivariate Beta models instead of approximating
numerical integration.

In this article, the following notation convention is adopted:
$W=\begin{pmatrix} X & Z \\ Z & Y \end{pmatrix}$ is the symmetric random
matrix to be considered. Its density function is given by Equation \eqref{eq:distr}, which can
also be treated as the joint density function of $X,Y,Z$.
$\abs{W}=XY-Z^2$.
Let $\E_{\alpha,\beta}[f(X,Y, Z)] = \int f(x,y,z)p(w)dw$ denotes the expectation
with $\mathbf{B}(\alpha, \beta;I_2)$ where $f(\cdot, \cdot, \cdot)$ is an arbitrary function with three
variables. We will compute $\E_{\alpha,\beta}[f(X,Y, Z)]$
when $f(X,Y,Z)$ takes the monomial form: $f(X,Y,Z)=X^m Y^r Z^{2t}$.

\section{Marginal Distribution}
In this section we will compute $\E_{\alpha,\beta}[f(X,Y, Z)]$
for $f(X,Y,Z)=X^m$ and show that $X$ is one dimensional Beta distribution.
To accomplish our goals, we need the following lemma:
\begin{lemma}\label{lem:AB}
	Let $A = XY - Z^2, B = 1 - X - Y + A$, then we have
	\begin{align}
	\E_{\alpha, \beta}[Af(X,Y,Z)] =&
	\frac{\alpha(\alpha-1/2)}{(\alpha+\beta)(\alpha+\beta-1/2)}\E_{\alpha+1, \beta}[f(X,Y,Z)]
	\label{eq:Aexp} \\
	\E_{\alpha,\beta}[Bf(X,Y,Z)] =&
	\frac{\beta(\beta-1/2)}{(\alpha+\beta)(\alpha+\beta-1/2)}\E_{\alpha, \beta+1}[f(X,Y,Z)]
	\label{eq:Bexp}
	\end{align}
\end{lemma}
\begin{proof}
	For multivariate Beta function we have
	$B_p(a, b) = \frac{\Gamma_p(a)\Gamma_p(b)}{\Gamma_p(a+b)}$
	where $\Gamma_p$ is the multivariate Gamma function (see \cite{ingham_1933}).
	For $p=2$ we have $\Gamma_2(a) = \sqrt{\pi}\Gamma(a)\Gamma(a-1/2)$.
	\begin{align*}
	\frac{\E_{\alpha, \beta}[Af(X,Y,Z)]}{\E_{\alpha+1, \beta}[f(X,Y,Z)]} &
	=\frac{B_2(\alpha+1,\beta)}{B_2(\alpha,\beta)}\\
	&=\frac{\Gamma_2(\alpha+1)}{\Gamma_2(\alpha)}
	\frac{\Gamma_2(\alpha+\beta)}{\Gamma_2(\alpha+\beta+1)}\\
	& =\frac{\Gamma(\alpha+1)}{\Gamma(\alpha)}
	\frac{\Gamma(\alpha+1/2)}{\Gamma(\alpha-1/2)}
	\frac{\Gamma(\alpha+\beta)}{\Gamma(\alpha+\beta+1)}
	\frac{\Gamma(\alpha+\beta-1/2)}{\Gamma(\alpha+\beta+1/2)}\\
	&=\frac{\alpha(\alpha-1/2)}{(\alpha+\beta)(\alpha+\beta-1/2)}
	\end{align*}
	Thus Equation \eqref{eq:Aexp} is proved and Equation \eqref{eq:Bexp} follows similarly.
\end{proof}
Using the above Lemma, we give the main conclusion of this section:
\begin{theorem}\label{thm:Xm}
	$\E_{\alpha, \beta}[X^m] =
	\prod_{i=0}^{m-1}\frac{\alpha+i}{\alpha+\beta+i}$, and $X$
	follows Beta distribution $\Beta(\alpha, \beta)$.
\end{theorem}
\begin{proof}
	Since the position of $X$ and $Y$ is symmetric,
	$\E_{\alpha, \beta}[X]=\E_{\alpha, \beta}[Y]$.
	Taking the expectation about $\Beta_2(\alpha, \beta)$
	on both sides of $B=1-X-Y+A$ and using the
	conclusion of Lemma \ref{lem:AB}, we have
	\begin{equation*}
	\frac{\beta(\beta-1/2)}{(\alpha+\beta)(\alpha+\beta-1/2)}
	= 1 - 2\E_{\alpha, \beta}[X] +
	\frac{\alpha(\alpha-1/2)}{(\alpha+\beta)(\alpha+\beta-1/2)}
	\end{equation*}
	Solving the about equation we get
	$\E_{\alpha, \beta}[X]=\frac{\alpha}{\alpha + \beta}$.
	Recursively using Equation \eqref{eq:Aexp} with $f(X,Y,Z)=X$
	we have $\E_{\alpha, \beta}[X^m] =
	\prod_{i=0}^{m-1}\frac{\alpha+i}{\alpha+\beta+i}$.
	This expression of moments
	is the same with that of Beta distribution on bounded interval $[0,1]$, we
	conclude that $X$ is actually Beta distribution $B(\alpha,
	\beta)$.
\end{proof}
\section{Mixed Moments}
In this section, we further compute $\E_{\alpha,\beta}[X^mY^rZ^{t}]$.
By symmetric property  $\E_{\alpha,\beta}[X^mY^rZ^{2t+1}] = 0$.
Therefore we only need to consider the case when the power of $Z$
is even. Firstly We consider the case when $r=0$:
\begin{theorem}\label{thm:mm}
	\begin{equation}\label{eq:ZXtm}
	\E_{\alpha, \beta}[X^mZ^{2t}] = \frac{(2t-1)!!}{2^t}
	\prod_{i=0}^{t-1}
	\frac{\beta+i}{(\alpha+\beta+i-1/2)}
	\frac{\prod_{i=0}^{t+m-1}\alpha+i}{\prod_{i=0}^{2t+m-1}
		\alpha+\beta+i}
	\end{equation}
\end{theorem}
\begin{proof}
	We use induction to show Equation \eqref{eq:ZXtm} is true.
	Firstly, Equation \eqref{eq:ZXtm} is true for $t=0$ from 
	Theorem \ref{thm:Xm}. Let $A, B$ be the same as those
	in Lemma \ref{lem:AB}.
	Suppose Equation \eqref{eq:ZXtm} holds for $\E[Z^{2t-2}X^m]$,
	using $Z^2=XY-A=X(1-X-A+B)-A$, then
	\begin{align*}
	\E_{\alpha,\beta}[Z^{2t}X^m] &=
	\E_{\alpha,\beta}[Z^{2t-2}X^m(X-X^2+AX-BX-A)] \\
	&= \E_{\alpha,\beta}[Z^{2t-2}(X^{m+1} - X^{m+2})] +
	\E_{\alpha,\beta}[A Z^{2t-2} (X^{m+1} - X^m)] \\
	&-
	\E_{\alpha,\beta}[BZ^{2t-2} X^{m+1}] \\
	&= \left(1-\frac{\alpha+t+m}{\alpha+\beta+2t+m-1}
	\right)\E_{\alpha, \beta}[Z^{2t-2}X^{m+1}] \\
	&+
	\frac{\alpha(\alpha-1/2)}{(\alpha+\beta)(\alpha+\beta-1/2)}
	\E_{\alpha+1,\beta}[Z^{2t-2}(X^{m+1} - X^m)] \\
	&-\frac{\beta(\beta-1/2)}{(\alpha+\beta)(\alpha+\beta-1/2)}
	\E_{\alpha,\beta+1}[Z^{2t-2}X^{m+1}]\\
	&=\frac{\beta+t-1}{\alpha+\beta + 2t+m-1}
	\E_{\alpha,\beta}[Z^{2t-2}X^{m+1}] \\
	&+\frac{(\alpha+t+m)(\alpha-1/2)}{(\alpha+\beta+2t-1+m)(\alpha+\beta + t - 3/2)}\\
	&\cdot
	\left(1-\frac{\alpha+\beta+2(t-1)+m+1}{\alpha+t+m}\right)
	\E_{\alpha, \beta}[Z^{2t-2}X^{m+1}] \\
	&-\frac{(\beta+t-1)(\beta-1/2)}{(\alpha+\beta+2t-1+m)(\alpha+\beta + t - 3/2)}
	\E_{\alpha,\beta}[Z^{2t-2}X^{m+1}]\\
	&=\frac{(t-1/2)(\beta+t-1)}{(\alpha+\beta+t-3/2)(\alpha+\beta+2t+m-1)}
	\E_{\alpha, \beta}[Z^{2t-2}X^{m+1}]
	\end{align*}
	Using Equation \eqref{eq:ZXtm} for $t-1$
	we can get the same form of expression for $t$.
\end{proof}
From Theorem \ref{thm:mm},
we can get
the general formula for the mixed moment when $m\geq r$:
\begin{cor}\label{cor:mr}
	\begin{align}
	\E_{\alpha, \beta}[X^mY^rZ^{2t}] &= \frac{(2t-1)!!}{2^t}
	\frac{\prod_{j=0}^{t-1} (\beta+j)}{\prod_{j=0}^{t+r-1} (\alpha+\beta-1/2+j)}
	\frac{\prod_{j=0}^{t+m-1}(\alpha+j)}{\prod_{j=0}^{2t+m-1} (\alpha+\beta+j)}\notag\\
	&\cdot\sum_{i=0}^r \frac{1}{2^i}\binom{r}{i}\prod_{j=1}^i (2t-1+2j)
	\frac{\displaystyle\prod_{j=0}^{r-i-1}(\alpha-1/2+j)
		\prod_{j=0}^{i-1}(\beta+j+t)}{\prod_{j=0}^{i-1}(\alpha+\beta+j+2t+m)}
	\label{eq:mrt}
	\end{align}
\end{cor}
Since $\E_{\alpha, \beta}[X^mY^rZ^{2t}]=\E_{\alpha, \beta}[X^rY^mZ^{2t}]$, 
when $m<r$ we can exchange $m$ with $r$ and then use Corollary \ref{cor:mr}.
\begin{proof}
	From Lemma \ref{lem:AB}, we have
	$X^m Y^r Z^{2t} = X^{m-r}(A+Z^2)^r Z^{2t} $.
	Using binomial theorem we have
	$\E_{\alpha, \beta}[X^m Y^r Z^{2t}] = \sum_{i=0}^r \binom{r}{i}\E_{\alpha, \beta}
	[A^{r-i}X^{m-r}Z^{2(t+i)}]$. Then using Equation \eqref{eq:Aexp} recursively
	we have 
	\begin{align*}
	\E_{\alpha, \beta}
	[A^{r-i}X^{m-r}Z^{2(t+i)}] =&
	\prod_{j=0}^{r-i-1}\frac{(\alpha+j)(\alpha+j-1/2)}
	{(\alpha+j+\beta)(\alpha+j+\beta-1/2)} \\
	\cdot & \E_{\alpha+r-i, \beta}
	[X^{m-r}Z^{2(t+i)}]
	\end{align*}
	Using Theorem \ref{thm:mm} we can finally get the expression in
	Equation \eqref{eq:mrt}.
\end{proof}
\section{Case Study}
In this section, we will give a natural example which illustrates how our result can be used.
We will consider the random matrix $S=QQ^T$ where $Q$ is $n\times k$ random orthogonal matrix.
We are interested in how $\E[S_{11}^mS_{12}^{2t}]$ changes as $k\to \infty$ when $r=\frac{k}{n}$ is fixed.

From Proposition 7.2 of \cite{eaton1989group},
$
\begin{pmatrix}
S_{11} & S_{12} \\
S_{21} & S_{22}
\end{pmatrix}
$
is exactly 2 by 2 random matrix of Beta distribution with parameter $B(\frac{k}{2}, \frac{n-k}{2}, I_2)$.
Using Theorem \ref{thm:mm}, we could write
$\E[S_{11}^mS_{12}^{2t}] \sim \frac{(2t-1)!!}{2^t} \frac{r^t(1-t)^{t+m}}{n^t} = O(n^{-t})$.
That is, $\E[S_{11}^mS_{12}^{2t}]$ decreases in the order of $n^{-t}$.

\section{Conclusion}
We have derived the formula of moments for multivariate Beta distribution
of 2 by 2 matrix. This result is helpful for analyzing other statistical properties of multivariate Beta distribution.
\bibliographystyle{plain}
\bibliography{exportlist}

\end{document}